\title{Optimal whitespace synchronization strategies}
\author{{Yossi Azar\thanks{Department of Computer Science, Tel-Aviv University, Tel-Aviv 69978, Israel. Email address: {\tt azar@tau.ac.il}.}} \qquad
{Ori Gurel-Gurevich\thanks{Microsoft
Research, One Microsoft Way, Redmond, WA 98052, USA. Email:
{\tt origurel@microsoft.com}.}} \qquad
{Eyal Lubetzky\thanks{Microsoft
Research, One Microsoft Way, Redmond, WA 98052, USA. Email:
{\tt eyal@microsoft.com}.}} \qquad
{Thomas Moscibroda\thanks{Microsoft
Research, One Microsoft Way, Redmond, WA 98052, USA. Email:
{\tt moscitho@microsoft.com}.}}}
\pdfoutput=1

\documentclass [11pt]{article}
\usepackage[]{amsmath,amssymb,amsfonts,latexsym,amsthm,enumerate,paralist,url}
\usepackage{amssymb,graphicx}
\usepackage[numeric,initials,nobysame]{amsrefs}
\usepackage[colorlinks=true, pdfstartview=FitV, linkcolor=blue, citecolor=blue, urlcolor=blue,pagebackref=false]{hyperref}

\date{}

\numberwithin{equation}{section}

\newtheorem{maintheorem}{Theorem}

\newtheorem{theorem}{Theorem}[section]
\newtheorem*{theorem*}{Theorem}

\newtheorem*{observation*}{Observation}
\newtheorem{corollary}[theorem]{Corollary}

\theoremstyle{definition}{

\newtheorem*{example*}{Example}

\newtheorem{remark}[theorem]{Remark}
\newtheorem*{remark*}{Remark}
}

\newcommand{\N}{\mathbb N}

\newcommand{\E}{\mathbb{E}}
\renewcommand{\P}{\mathbb{P}}

\newcommand{\Geom}{\operatorname{Geom}}

\renewcommand{\epsilon}{\varepsilon}

\begin{document}
\maketitle

\begin{abstract}
The whitespace-discovery problem describes two parties, Alice and Bob, trying to establish a communication channel over one of a given large segment of whitespace channels.
Subsets of the channels are occupied in each of the local environments surrounding Alice and Bob, as well as in the global environment between them (Eve). In the absence of a common clock for the two parties, the goal is to devise time-invariant (stationary) strategies minimizing the synchronization time.
This emerged from recent applications in discovery of wireless devices.

We model the problem as follows. There are $N$ channels, each of which is open (unoccupied) with probability $p_1,p_2,q$ independently for Alice, Bob and Eve respectively. Further assume that $N \gg 1/(p_1 p_2 q)$ to allow for sufficiently many open channels. Both Alice and Bob can detect which channels are locally open and every time-slot each of them chooses one such
channel for an attempted sync. One aims for strategies that, with high probability over the environments,
guarantee a shortest possible expected sync time depending only on the $p_i$'s and $q$.

Here we provide a stationary strategy for Alice and Bob with a guaranteed expected sync time of $O(1 / (p_1 p_2 q^2))$ given that each party also has knowledge of $p_1,p_2,q$. When the parties are oblivious of these probabilities, analogous strategies incur a cost of a poly-log factor, i.e.\ $\tilde{O}(1 / (p_1 p_2 q^2))$.
Furthermore, this performance guarantee is essentially optimal as we show that any stationary strategies of Alice and Bob have an expected sync time of at least $\Omega(1/(p_1 p_2 q^2))$.
\end{abstract}


\section{Introduction}\label{sec:intro}

Consider two parties, Alice and Bob, who wish to establish a communication channel in one out of a segment of $N$ possible channels. Subsets of these channels may already be occupied in the local environments of either Alice or Bob, as well as in the global environment in between them whose users are denoted by Eve. Furthermore, the two parties do not share a common clock and hence one does not know for how long (if at all) the other party has already been trying to communicate. Motivated by applications in discovery of wireless devices, the goal is thus to devise time-invariant strategies that ensure fast synchronization with high probability (w.h.p.) over the environments.

We formalize the above problem as follows.
Transmissions between Alice and Bob go over three environments: local ones around Alice and Bob and an additional global one in between them, Eve. Let $A_i,B_i,E_i$ for $i=1,\ldots,N$ be the indicators for whether a given channel is open (unoccupied) in the respective environment. Using local diagnostics Alice knows $A$ yet does not know $B,E$ and analogously Bob knows $B$ but is oblivious of $A,E$. In each time-slot, each party selects a channel to attempt communication on (the environments do not change between time slots). The parties are said to \emph{synchronize} once they select the same channel $i$ that happens to be open in all environments (i.e., $A_i=B_i=E_i=1$). The objective of Alice and Bob is to devise strategies that would minimize their expected synchronization time.

For a concrete setup, let $A_i,B_i,E_i$ be independent Bernoulli variables with probabilities $p_1,p_2,q$ respectively for all $i$, different channels being independent of each other. (In some applications the two parties have knowledge of the environment densities $p_1,p_2,q$ while in others these are unknown.) Alice and Bob then seek strategies whose expected sync time over the environments is minimal.

\begin{example*}
  Suppose that $p_1=p_2=1$ (local environments are fully open) and Alice and Bob use the naive strategy of selecting a channel uniformly over $[N]$ and independently every round. If there are $Q \approx q N$ open channels in the global environment Eve then the probability of syncing in a given round is $Q/N^2 \approx q/N$, implying an expected sync time of about $N / q$ to the very least.
\end{example*}

In the above framework it could occur that \emph{all} channels are closed, in which case the parties can never sync; as a result, unless this event is excluded the expected sync time is always infinite.
However, since this event has probability at most $(1-p_1 p_2 q)^N \leq \exp(-N p_1 p_2 q)$ it poses no real problem for applications (described in further details later) where $N \gg 1/(p_1 p_2 q)$. In fact, we aim for performance guarantees that depend only on $p_1,p_2,q$ rather than on $N$, hence a natural way to resolve this issue is to extend the set of channels to be infinite, i.e.\ define $A_i,B_i,E_i$ for every $i\in \N$. (Our results can easily be translated to the finite setting with the appropriate exponential error probabilities.)

A \emph{strategy} is a sequence of probability measures $\{\mu^t\}$ over $\N$, corresponding to a randomized choice of channel for each time-slot $t\geq 1$. Suppose that Alice begins the discovery via the strategy $\mu_a$ whereas Bob begins the synchronization attempt at time $s$ via the strategy $\mu_b$.
Let $X_t$ be the indicator for a successful sync at time $t$ and let $X$ be the first time Alice and Bob sync, that is
\begin{align}
\P(X_t = 1 \mid A,B,E) &= \sum_j \mu_a^t(j) \mu_b^{s+t}(j) A_j B_j E_j\,,\label{eq-Xt-def}\\
X &= \min\{ t : X_t = 1\}\,.\label{eq-X-def}
\end{align}
The choice of $\mu_a,\mu_b$ aims to minimize $\E X$ where the expectation is over $A,B,E$ as well as the randomness of Alice and Bob in applying the strategies $\mu_a,\mu_b$.
\begin{example*}[\emph{fixed strategies}]
  Suppose that both Alice and Bob apply the same pair of strategies independently for all rounds, $\mu_a$ and $\mu_b$ respectively. In this special case, given the environments $A,B,E$ the random variable $X$ is geometric with success probability $\sum_j \mu_a(j) \mu_b(j) A_j B_j E_j$, thus the mappings $A\mapsto \mu_a$ and $B\mapsto \mu_b$ should minimize $\E X = \E \left[\big(\sum_j \mu_a(j) \mu_b(j) A_j B_j E_j\big)^{-1} \right]$.
\end{example*}

A crucial fact in our setup is that Alice and Bob have no common clock and no means of telling whether or not their peer is already attempting to communicate (until they eventually synchronize). As such, they are forced to apply a \emph{stationary} strategy, where the law at each time-slot is identical (i.e.\ $\mu^t \sim \mu^1$ for all $t$).
For instance, Alice may choose a single $\mu_a$ and apply it independently in each step (cf.\ above example). Alternatively, strategies of different time-slots can be highly dependent, e.g.\ Bob may apply a periodic policy comprising $\mu_b^1,\ldots,\mu_b^n$ and a uniform initial state $s\in[n]$.

The following argument demonstrates that stationary strategies are essentially optimal when there is no common clock between the parties. Suppose that Alice has some finite (arbitrarily long) sequence of strategies $\{\mu_a^t\}_1^{M_a}$ and similarly Bob has a sequence of strategies $\{\mu_b^t\}_1^{M_b}$. With no feedback until any actual synchronization we may assume that the strategies are non-adaptive, i.e.\ the sequences are determined in advance. Without loss of generality Alice is joining the transmission after Bob has already attempted some $\beta$ rounds of communication, in which case the expected synchronization time is $\E_{0,\beta} X$, where $\E_{\alpha,\beta} X$ denotes the expectation of $X$ as defined in~\eqref{eq-Xt-def},\eqref{eq-X-def} using the strategies $\{\mu_a^{t+\alpha}\},\{\mu_b^{t+\beta}\}$. Having no common clock implies that in the worst case scenario (over the state of Bob) the expected time to sync is $\max_{\beta} \E_{0,\beta} X$ and it now follows that Bob is better off modifying his strategy into a stationary one by selecting $\beta \in [M_b]$ uniformly at random, leading to an expected synchronization time of $ M_b^{-1} \sum_{\beta}\E_{0,\beta} X$.

\subsection{Optimal synchronization strategies}
Our main result is a recipe for Alice and Bob to devise stationary strategies guaranteeing an optimal expected synchronization time up to an absolute constant factor, assuming they know the environment densities $p_1,p_2,q$ (otherwise the expected sync time is optimal up to a poly-log factor).

\begin{maintheorem}\label{thm-1}
Consider the synchronization problem with probabilities $p_1,p_2,q$ for the environments $A,B,E$ respectively, and let $X$ denote the expected sync time. The following then holds:
\begin{compactenum}[(i)]
  \item \label{thm-1-upper}
  There are fixed strategies for Alice and Bob guaranteeing $\E X = O(1/(p_1 p_2 q^2))$, namely:
  \begin{compactitem}
    \item Alice takes $\mu_a \sim \Geom(p_2 q/6)$ over her open channels $\{i : A_i=1\}$,
    \item Bob takes $\mu_b \sim \Geom(p_1 q/6)$ over his open channels $\{i : B_i = 1\}$.
  \end{compactitem}
  Furthermore, for any fixed $\epsilon > 0$ there are fixed strategies for Alice and Bob that do not require knowledge of $p_1,p_2,q$ and guarantee $\E X = O\Big(\frac{1}{p_1 p_2 q^2}\log^{2+\epsilon}\big(\frac{1}{p_1 p_2 q}\big)\Big) = \tilde{O}\big(\frac{1}{p_1 p_2 q^2}\big)$, obtained by taking $\mu_a(\mbox{$j$-th open $A$ channel}) = \mu_b(\mbox{$j$-th open $B$ channel}) \propto 1/(j\log^{1+\epsilon/2} j)$.
  \item \label{thm-1-lower} The above strategies are essentially optimal as every possible choice of stationary strategies by Alice and Bob satisfies $\E X = \Omega(1/(p_1 p_2 q^2))$.
\end{compactenum}
\end{maintheorem}

\begin{remark*}
The factor $1/6$ in the parameters of the geometric distributions can be fine-tuned to any smaller (or even slightly larger) fixed $\alpha>0$ affecting the overall expected sync time $\E X$ by a multiplicative constant. See Fig.~\ref{fig:sync} for a numerical evaluation of $\E X$ for various values of $\alpha$.
\end{remark*}

Recall that Alice and Bob must apply stationary strategies in the absence of any common clock or external synchronization device shared by them, a restriction which is essential in many of the applications of wireless discovery protocols. However, whenever a common external clock does happen to be available there may be strategies that achieve improved performance. The next theorem establishes the optimal strategies in this simpler scenario.

\begin{maintheorem} \label{thm-2}
Consider the synchronization problem with probabilities $p_1,p_2,q$ for the environments $A,B,E$ respectively, and let $X$ denote the expected sync time. If Alice and Bob have access to a common clock then there are non-stationary strategies for them achieving $\E X =  O(1/(\min\{p_1,p_2\} q))$. Moreover, this is tight as the expected sync time is always $\Omega(1/(\min\{p_1,p_2\}q))$.
\end{maintheorem}

\begin{figure}[t]
\centering
\fbox{\includegraphics[width=5in]{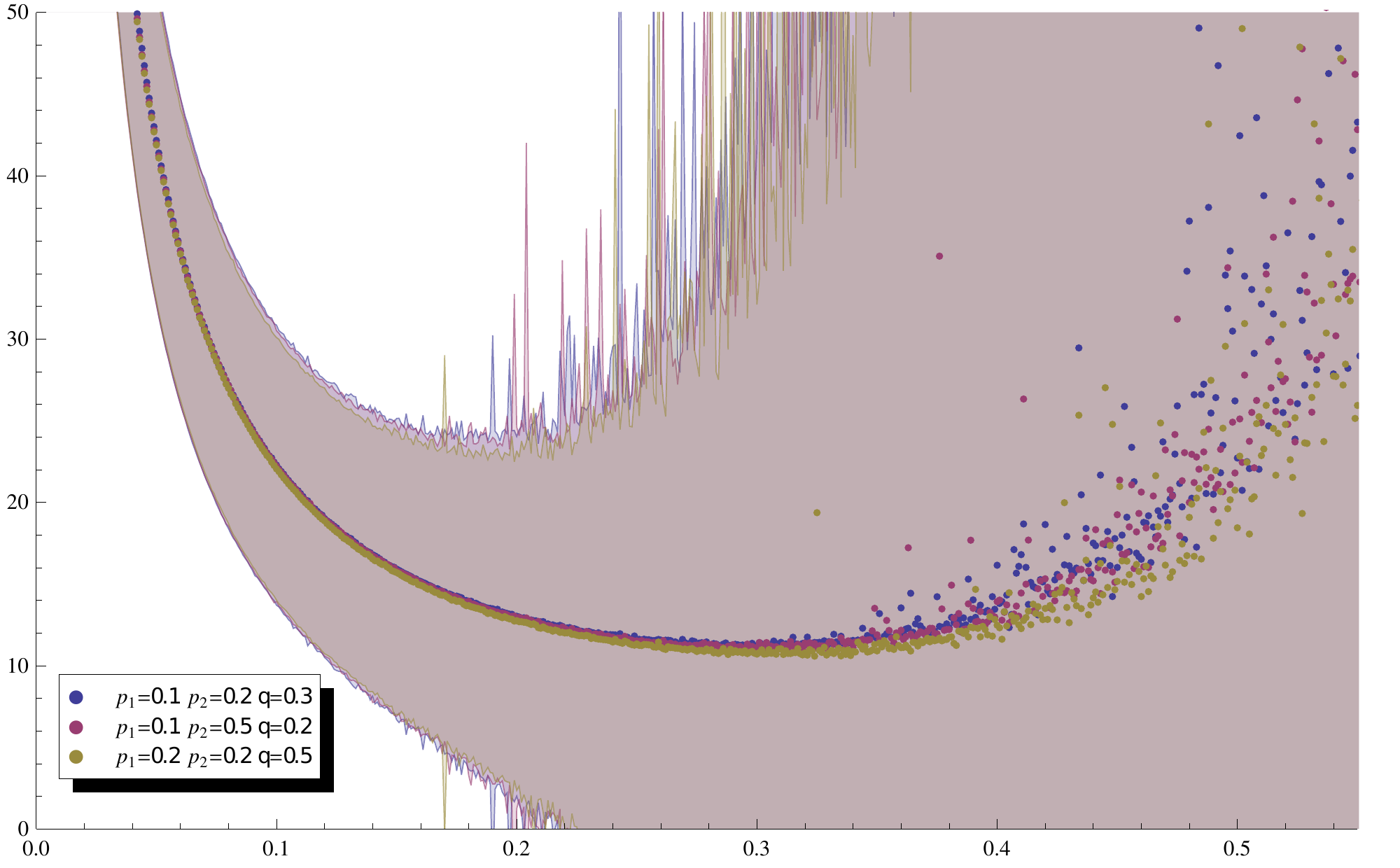}}
\caption{Synchronization time $\E X $ as in~\eqref{eq-X-def} normalized by a factor of $p_1 p_2 q^2$ for a protocol using geometric distributions with parameters $\alpha p_i q$ for various values of $0<\alpha<1$. Markers represent the average of the expected synchronization time $\E X$ over $10^5$ random environments with $n=10^4$ channels; surrounding envelopes represent a window of one standard deviation around the mean.}
\label{fig:sync}
\end{figure}

\subsection{Applications in wireless networking and related work}

The motivating application for this work comes from recent developments in wireless networking. In late 2008, the FCC issued a historic ruling permitting the use of unlicensed devices in the unused portions of the UHF spectrum (mainly the part between 512Mhz and 698Mhz), popularly referred to as ``whitespaces''. Due to the potential for substantial bandwidth and long transmission ranges, whitespace networks (which are, in a more general context, also frequently referred to as cognitive wireless networks) represent a tremendous opportunity for mobile and wireless communication. One critical equipment imposed by the FCC in its ruling is that whitespace wireless devices must not interfere with incumbents, i.e., the current users of this spectrum (specifically, in the UHF bands, these are TV broadcasters as well as wireless microphones). Hence, these incumbents are considered ``primary users'' of the spectrum, while whitespace devices are secondary users and are allowed to use the spectrum only opportunistically,  whenever no primary user is using it (The FCC mandates whitespace devices to detect the presence of primary users using a combination of sensing techniques and a geo-location database). At any given time, each whitespace device thus has a spectrum map on which some parts are blocked off while others are free to use.

The problem studied in this paper captures exactly the situation in whitespace networks when two nodes $A$ and $B$ seek to discover one another to establish a connection. Each node knows its own free channels on which it can transmit, but it does not know which of these channels may be available at the other node, too. Furthermore, given the larger transmission range in whitespace networks (up to a mile at Wi-Fi transmission power levels), it is likely that the spectrum maps at $A$ and $B$ are similar yet different. For example, a TV broadcast tower is likely to block off a channel for both $A$ and $B$, but a wireless microphone --- due to its small transmission power --- will prevent only one of the nodes from using a channel.

Thus far, the problem of synchronizing/discovery of whitespace nodes has only been addressed
when one of the nodes is a fixed access point (AP) and the other node is a client. Namely, in the framework studied in~\cite{BCMMW} the AP broadcasts on a fixed channel and the client node wishes to scan its local environment and locate this channel efficiently. That setting thus calls for technological solutions (e.g.\ based on
scanning wider channel widths) to allow the client
to find the AP channel faster than the approach of searching all possible channels one by one.

To the best of our knowledge, the results in this paper are the first to provide an efficient synchronization scheme in the  setting where both nodes are remote clients that may broadcast on any given channel in the whitespace region.

\subsection{Related work on Rendezvous games}

From a mathematical standpoint, the synchronization problems considered in this paper seem to belong to the field of Rendezvous Search Games. The most familiar problem of this type is known as The Telephone Problem or The Telephone Coordination Game. In the telephone problem each of two players is placed in a distinct room with $n$ telephone lines connecting the rooms. The lines are not labeled and so the players, who wish to communicate with each other, cannot simply use the first line (note that, in comparison, in our setting the channels are labeled and the difficulty in synchronizing is due to the local and global noise).

The optimal strategy in this case, achieving an expectation of $n/2$, is for the first player to pick a random line and continue using it, whereas the second player picks a uniformly random permutation on the lines and try them one by one. However, this strategy requires the players to determine which is the first and which is the second. It is very plausible that such coordination is not possible, in which case we require both players to employ the same strategy.

The obvious solution is for each of them to pick a random line at each turn, which gives an expectation of $n$ turns. It turns out, however, that there are better solutions: Anderson and Weber~\cite{AW} give a solution yielding an expectation of $\approx 0.8288497 n$ and conjecture it's optimality.

To our knowledge, the two most prominent aspects of our setting, the presence of asymmetric information and the stationarity requirement (stemming from unknown start times) have not been considered in the literature. For example, the Anderson-Weber strategy for the telephone problem is not stationary --- it has a period of $n-1$. It would be interesting to see what can be said about the optimal stationary strategies for this and other rendezvous problems. The interested reader is referred to~\cites{ABE,AG} and the references therein for more information on rendezvous search games.

\section{Analysis of synchronization strategies}\label{sec:main}

\subsection{Proof of Theorem~\ref{thm-1}, upper bound on the sync time}

Let $\mu_a$ be geometric with mean $(\alpha p_{2}q)^{-1}$ over the open channels for Alice $\{i:A_i=1\}$ and analogously let $\mu_b$ be geometric with mean $(\alpha p_{1}q)^{-1}$ over the open channels for Bob $\{i:B_i=1\}$, where $0<\alpha<1$ will be determined later.

Let $J = \min\{ j : A_j = B_j = E_j = 1\}$ be the minimal channel open in all three environments. Further let $J_a,J_b$ denote the number of locally open channels prior to channel $J$ for Alice and Bob resp., that is
\begin{align*}
& J_a = \#\{ j < J : A_j = 1 \}\,,\quad J_b = \#\{ j < J : B_j = 1 \}\,.
\end{align*}
Finally, for some integer $k \geq 0$ let $M_k$ denote the event
\begin{equation}\label{eq-max-J0-J1}
  k \leq \max\left\{ J_a p_2 q ~,~ J_b p_1 q\right\} < k + 1 \,.
\end{equation}
Notice that, by definition, Alice gives probability $(1 - \alpha p_2 q)^{j-1} \alpha p_2 q$ to her $j$-th open channel while Bob gives probability $(1 - \alpha p_1 q)^{j-1} \alpha p_1 q$ to his $j$-th open channel. Therefore, on the event $M_k$ we have that in any specific round, channel $J$ is chosen by both players with probability at least
\begin{align*}
(1 - \alpha p_1 q)^{\frac{k+1}{p_1 q}} (1- \alpha p_2 q)^{\frac{k+1}{p_2 q}} \alpha^2 p_1 p_2 q^2 \geq
\mathrm{e}^{-4\alpha(k+1)} \alpha^2 p_1 p_2 q^2\,,
\end{align*}
where in the last inequality we used the fact that $(1-x)\geq \exp(-2x)$ for all $0 \leq x\leq \frac12$, which will be justified by later choosing $\alpha < \frac12$.
Therefore, if $X$ denotes the expected number of rounds required for synchronization, then
\begin{align}
  \label{eq-exp-on-Mk}
  \E[ X \mid M_k ] \leq \mathrm{e}^{4\alpha(k+1)} (\alpha^2 p_1 p_2 q^2)^{-1}\,.
\end{align}
On the other hand, $J_a$ is precisely a geometric variable with the rule $\P(J_a=j) = (1-p_2 q)^j p_2 q$ and similarly $\P(J_b=j) = (1-p_1 q)^j p_1 q$. Hence,
\[ \P(M_k) \leq (1-p_2 q)^{k/(p_2 q)} + (1-p_1 q)^{k/(p_1 q)} \leq 2\mathrm{e}^{-k}\,.\]
Combining this with~\eqref{eq-exp-on-Mk} we deduce that
\begin{equation}
   \label{eq-Ex-upper-bound}
   \E X \leq 2 \sum_k \mathrm{e}^{-k} \E[X \mid M_k] \leq 2\mathrm{e}^{4\alpha} (\alpha^2 p_1 p_2 q^2)^{-1} \sum_k
\mathrm{e}^{(4\alpha-1)k} \leq \frac{2\mathrm{e}}{\alpha^2\left(\mathrm{e}^{1-4\alpha}-1\right)}\left(p_1 p_2 q^2\right)^{-1}
 \end{equation}
where the last inequality holds for any fixed $\alpha < \frac14$. In particular, a choice of $\alpha = \frac16$ implies that $\E X \leq 500/ \left(p_1 p_2 q^2\right)$, as required.
\qed

\begin{remark}\label{rem-p1=p2}
  In the special case where $p_1=p_2$ (denoting this probability simply by $p$) one can optimize the choice of constants in the proof above to obtain an upper bound of $\E X \leq 27 / (p q)^2$.
\end{remark}

\subsection{Strategies oblivious of the environment densities}
Observe first that if Alice and Bob multiply the parameters of their geometric distributions (as specified in Part~\eqref{thm-1-upper} of Theorem~\ref{thm-1}) by some absolute constant $0 < c < 1$ then the upper bound~\eqref{eq-Ex-upper-bound} on $\E X$ will increase by a factor of at most some absolute $C > 0$.

With this in mind, fix $\epsilon > 0$ and consider the following strategies.
Every round, Alice guesses $p_2 q$ to be $\exp(-i)$ for $i=1,2,\ldots$ with probability proportional to $i^{-(1+\epsilon/2)}$, while Bob does the same for $p_1 q$. This way, both Alice and Bob successfully guess these parameters (within a factor of $\mathrm{e}$) with probability at least
\[ c \log^{-(1+\epsilon/2)}\Big(\frac{1}{p_1 q}\Big) \log^{-(1+\epsilon/2)}\Big(\frac{1}{p_2 q}\Big) \geq c \log^{-(2+\epsilon)} \Big(\frac{1}{p_1 p_2 q}\Big)\]
where $c > 0$ is some absolute constant. Hence, restricting our analysis to these rounds only we obtain an expected sync time of at most $O(1/(p_1 p_2 q^2))$ such rounds and overall
 \[ \E X = O\Big(\frac{1}{p_1 p_2 q^2} \log^{2+\epsilon}\Big(\frac{1}{p_1 p_2 q}\Big)\Big) = \tilde{O}\Big(\frac{1}{p_1 p_2 q^2}\Big)\,.\]

An elementary calculation shows that the above strategy is equivalent to having both Alice and Bob choose their corresponding $j$-th open channel ($j=1,2,\ldots$) with probability proportional to $j \log^{1+\epsilon/2}j $. Indeed, one may repeat the arguments from the previous section directly for these strategies and obtain that (in the original notation) for any given round
\[ \P(\mbox{Alice and Bob select $J$}\mid M_k) \geq c p_1 p_2 q^2 \left[ (k+1)^2 \log^{1+\epsilon/2}\Big(\frac{k+1}{p_1 q}\Big)\log^{1+\epsilon/2}\Big(\frac{k+1}{p_2 q}\Big) \right]^{-1}\]
for some absolute constant $c>0$. From this we then infer that
\begin{align*}
\E X &\leq O\left(1/(p_1 p_2 q^2)\right) \sum_k \mathrm{e}^{-k} (k+1)^2 \log^{1+\epsilon/2}\Big(\frac{k+1}{p_1 q}\Big)\log^{1+\epsilon/2}\Big(\frac{k+1}{p_2 q}\Big) \\
&= O\Big(\frac{1}{p_1 p_2 q^2} \log^{2+\epsilon}\Big(\frac{1}{p_1 p_2 q}\Big)\Big)\,,
\end{align*}
as argued above. \qed

\subsection{Proof of Theorem~\ref{thm-1}, lower bound on the sync time}

\begin{theorem}\label{thm-R-lower-bound}
Let $\mu_a,\mu_b$ be the stationary distribution of the strategies of Alice and Bob resp., and let $R=\sum_j \mu_a(j) \mu_b(j) A_j B_j E_j$ be the probability of successfully syncing in any specific round. Then there exists some absolute constant $C>0$ such that $\P(R<C p_0 p_1 q^2)\ge \frac12$.
\end{theorem}
\begin{proof}
Given the environments $A,B$ define
\begin{align*}
S^a_k = \{j \,:\, 2^{-k} < \mu_a(j) \le 2^{-k+1} \} \,,\quad
S^b_k = \{j \,:\, 2^{-k} < \mu_b(j) \le 2^{-k+1} \} \,.
\end{align*}
Notice that the variables $S^a_k$ are a function of the strategy of Alice which in turn depends on her  local environment $A$ (an analogous statement holds for $S^b_k$ and $B$). Further note that clearly
\[ |S^a_k| < 2^k\quad\mbox{ and }\quad |S^b_k| < 2^k\quad\mbox{ for any $k$}\,.\]
Let $T^a_k$ denote all the channels where the environments excluding Alice's (i.e., both of the other environments $B,E$) are open, and similarly let $T^b_k$ denote the analogous quantity for Bob:
\[ T^a_k = \{j\in S^a_k \,:\, B_j=E_j=1 \}\,,\quad
 T^b_k = \{j\in S^b_k \,:\, A_j=E_j=1 \}\,. \]
Obviously, $\E|T^a_k| < 2^k p_2 q$ and $\E|T^b_k| < 2^k p_1 q$.

Since $\{B_j\}_{j\in \N}$ and $\{E_j\}_{j\in \N}$ are independent of $S^a_k$ (and of each other), for any $\beta>0$ we can use the Chernoff bound (see, e.g., \cite{JLR}*{Theorem 2.1} and~\cite{AS}*{Appendix A}) with a deviation of $t = (\beta-1)2^k p_2 q$ from the expectation to get
\[ \P\left(|T^a_k|>\beta 2^k p_2 q \right) < \exp\bigg(- \frac32 \frac{(\beta-1)^2}{\beta+2} 2^k p_2 q\bigg) \,,\]
and analogously for Bob we have
\[ \P\left(|T^b_k|>\beta 2^k p_1 q \right) < \exp\bigg(- \frac32 \frac{(\beta-1)^2}{\beta+2} 2^k p_1 q\bigg) \,.\]
Clearly, setting
\[ K_a = \log_2(1/(p_2 q))-3 \,,\quad K_b = \log_2(1/(p_1 q))-3 \] and taking $\beta$ large enough (e.g., $\beta=20$ would suffice) we get
\begin{align}\label{eq-Tk-large1} \P\left( \mbox{$\bigcup$}_{k \geq K_a}\left\{ |T^a_k|> \beta 2^k p_2 q\right\} \right)
\leq 2\, \P\left(|T^a_{K_a}|>\beta 2^{K_a} p_2 q \right) < \frac18
\end{align}
and
\begin{align}
\label{eq-Tk-large2} \P\left( \mbox{$\bigcup$}_{k \geq K_b}\left\{ |T^b_k|> \beta 2^k p_1 q\right\} \right) < \frac18\,.
 \end{align}
Also, since $\sum_{k < K_a} |S^a_k| < 2^{K_a} \le (8 p_2 q)^{-1}$ and
similarly $\sum_{k < K_b} |S^b_k| < 2^{K_b} \le (8p_1 q)^{-1}$,
we have by Markov's inequality that
\begin{equation}
  \label{eq-Tk-empty1}
  \P\left( \mbox{$\bigcup$}_{k < K_a}\left\{ |T^a_k| > 0\right\} \right) \leq
\sum_{k<K_a} \E |T_k^a| = p_2 q \sum_{k<K_a} \E |S_k^a| < \frac18
\end{equation}
and similarly
\begin{equation}
  \label{eq-Tk-empty2}
  \P\left( \mbox{$\bigcup$}_{k < K_b}\left\{ |T^b_k| > 0\right\} \right) < \frac18\,.
\end{equation}
Putting together \eqref{eq-Tk-large1},\eqref{eq-Tk-large2},\eqref{eq-Tk-empty1},\eqref{eq-Tk-empty2}, with probability at least $\frac12$ the following holds:
\begin{equation}
  \label{eq-control-Tk}
  |T_k^a| \leq \left\{\begin{array}{ll}
  \beta 2^k p_2 q & k \geq K_a \\
  0 & k < K_a\end{array}\right.\,,\quad
  |T_k^b| \leq \left\{\begin{array}{ll}
  \beta 2^k p_1 q & k \geq K_b \\
  0 & k < K_b\end{array}\right.
  \quad\mbox{ for all $k$.}
\end{equation}
When \eqref{eq-control-Tk} holds we can bound $R$ as follows:
\begin{align*}
R&=\sum_j \mu_a(j) \mu_b(j) A_j B_j E_j =\sum_k \sum_\ell \sum_{j\in T^a_k \cap T^b_\ell} \mu_a(j) \mu_b(j)
\le \sum_k \sum_\ell |T^a_k \cap T^b_\ell| 2^{-k+1} 2^{-\ell+1} \\
&\le \sum_k \sum_\ell \sqrt{|T^a_k|\, |T^b_\ell|} 2^{-k+1} 2^{-\ell+1}  = 4\bigg(\sum_k \sqrt{|T_k^a|}2^{-k}\bigg) \bigg(\sum_\ell \sqrt{|T_\ell^b|}2^{-\ell}\bigg) \\
&\leq 4 \beta (p_1 p_2)^{1/2} q \bigg(\sum_{k\ge K_a} 2^{-k/2}\bigg)  \bigg(\sum_{\ell\geq K_b} 2^{-\ell/2}\bigg) \,,
\end{align*}
where the second inequality used the fact that $|F_1 \cap F_2| \leq \min\{|F_1|,|F_2|\}\leq \sqrt{|F_1||F_2|}$ for any two finite sets $F_1,F_2$ and
the last inequality applied~\eqref{eq-control-Tk}. From here the proof is concluded by observing that
\begin{align*}
 R &\leq 16 (p_1 p_2)^{1/2} q 2^{-K_a/2} 2^{-K_b/2} = 128 \beta p_1 p_2 q^2\,.\qedhere
\end{align*}

\end{proof}

\begin{corollary}\label{cor-stat-lower-bound}
There exists some absolute $c>0$ such that for any pair of stationary strategies, the expected number of rounds required for a successful synchronization is at least $c/(p_1 p_2 q^2)$.
\end{corollary}
\begin{proof}
Recall that the success probability in any specific round, given the environments and strategies, is precisely $R$. Hence, conditioned on the value of $R$, the probability of synchronizing in one of the first $1/(2 R)$ rounds is at most $\frac12$. Theorem~\ref{thm-R-lower-bound} established that with probability at least $\frac12$ we have $R < C p_1 p_2 q^2$, therefore altogether with probability at least $\frac14$ there is no synchronizing before time $(2Cp_1 p_2 q^2)^{-1}$. We conclude that the statement of the corollary holds with $c = 1/(8C)$.
\end{proof}

\subsection{Proof of Theorem~\ref{thm-2}}

Partition the channels into infinitely many sets of infinitely many channels each. On the $i$-th round, Alice chooses the first channel in the $i$-th set which is open in her environment. Bob does likewise.

Consider the probability that both parties choose the same channel: Each channel has probability $1-(1-p_1)(1-p_2)=p_1+p_2-p_1 p_2$ of being open to Alice or Bob and probability $p_1 p_2$ of being open to both. Hence, the probability that the first channel open to either is open to both is
$$\frac{p_1 p_2}{p_1+p_2 - p_1 p_2} \ge \min\{p_1, p_2\}/2 \, .$$

If indeed both players chose the same channel at some round, it is necessarily open for both of them and with probability $q$ it is also open in the global environment. Hence, the probability of success at each round is at least $\min\{p_1,p_2\}q/2$. Different rounds use disjoint sets of channels, so the event of success at different rounds are independent and the number of round to success has a geometric distribution with expectation $\E X \leq 2/(\min\{p_1,p_2\}q)$.

For a lower bound of matching order observe the following: For any strategy Alice might employ and at any given round, the probability that the channel she chooses is open for both Bob and Eve is $p_2 q$, and this is an upper bound for the probability of success. Similar argument for Bob yields that the probability of success at any given round is at most $\min\{p_1,p_2\}q$.
A straightforward first moment argument (as in the proof of Corollary~\ref{cor-stat-lower-bound}) now implies that $\E X\ge 1/(4\min\{p_1,p_2\}q)$.

This completes the proof. \qed




\begin{bibdiv}
\begin{biblist}[\normalsize]

\bib{AS}{book}{
  author={Alon, Noga},
  author={Spencer, Joel H.},
  title={The probabilistic method},
  edition={3},
  publisher={John Wiley \& Sons Inc.},
  place={Hoboken, NJ},
  date={2008},
  pages={xviii+352},
}


\bib{ABE}{article}{
   author={Alpern, Steve},
   author={Baston, V. J.},
   author={Essegaier, Skander},
   title={Rendezvous search on a graph},
   journal={J. Appl. Probab.},
   volume={36},
   date={1999},
   number={1},
   pages={223--231},
}
		
\bib{AG}{book}{
   author={Alpern, Steve},
   author={Gal, Shmuel},
   title={The theory of search games and rendezvous},
   series={International Series in Operations Research \& Management
   Science, 55},
   publisher={Kluwer Academic Publishers},
   place={Boston, MA},
   date={2003},
   pages={xvi+319},
}

\bib{AW}{article}{
   author={Anderson, E. J.},
   author={Weber, R. R.},
   title={The rendezvous problem on discrete locations},
   journal={J. Appl. Probab.},
   volume={27},
   date={1990},
   number={4},
   pages={839--851},
}


\bib{BCMMW}{article}{
  title={White space networking with wi-fi like connectivity},
  author={Bahl, P.},
  author={Chandra, R.},
  author={ Moscibroda, T.},
  author={Murty, R.},
  author={Welsh, M.},
  journal={ACM SIGCOMM Computer Communication Review},
  volume={39},
  pages={27--38},
  date={2009},
}

		

\bib{JLR}{book}{
   author={Janson, Svante},
   author={{\L}uczak, Tomasz},
   author={Rucinski, Andrzej},
   title={Random graphs},
   series={Wiley-Interscience Series in Discrete Mathematics and
   Optimization},
   publisher={Wiley-Interscience, New York},
   date={2000},
   pages={xii+333},
}


\end{biblist}
\end{bibdiv}
\end{document}